\documentclass[10pt,conference]{IEEEtran}
\usepackage{amsmath,amssymb}

\begin{document}

\newtheorem{definition}{Definition}[section]
\newtheorem{thm}{Theorem}[section]
\newtheorem{proposition}[thm]{Proposition}
\newtheorem{lemma}[thm]{Lemma}
\newtheorem{corollary}[thm]{Corollary}
\newtheorem{exam}{Example}[section]

\newtheorem{remark}{Remark}[section]

\newcommand{\La}{\mathbf{L}}
\newcommand{\h}{{\mathbf h}}
\newcommand{\Z}{{\mathbb Z}}
\newcommand{\R}{{\mathbb R}}
\newcommand{\C}{{\mathbb C}}
\newcommand{\D}{{\mathcal D}}
\newcommand{\F}{{\mathbf F}}
\newcommand{\HH}{{\mathbf H}}
\newcommand{\OO}{{\mathcal O}}
\newcommand{\G}{{\mathcal G}}
\newcommand{\A}{{\mathcal A}}
\newcommand{\B}{{\mathcal B}}
\newcommand{\I}{{\mathcal I}}
\newcommand{\E}{{\mathcal E}}
\newcommand{\PP}{{\mathcal P}}
\newcommand{\Q}{{\mathbb Q}}
\newcommand{\separ}{\,\vert\,}
\newcommand{\abs}[1]{\vert #1 \vert}
\newcommand{\mindet}[1]{\hbox{\rm det}_{min}\left( #1\right)}

\title{An algebraic look into MAC-DMT of lattice space-time codes}
\author{
\authorblockN{Roope Vehkalahti, \textit{Member, IEEE} }
\authorblockA{Department of Mathematics\\
University of Turku\\
Finland\\
Email: roiive@utu.fi}
\and
\authorblockN{Hsiao-feng (Francis) Lu,  \textit{Member, IEEE}}
\authorblockA{Department of Electronical Engineering \\
National Chiao Tung University\\
Hsinchu, Taiwan\\
Email:francis@cc.nctu.edu.tw }

}
\maketitle

\title{MAC-DMT in some asymmetric situations}

\begin{abstract}
In this paper we are concentrating on the diversity-multiplexing gain trade-off (DMT) of some space-time lattice codes. First we give a DMT bound
for lattice codes having restricted dimension. We then recover the well known results of the DMT of algebraic number field codes and the Alamouti code 
by using the union bound and see that these codes do achieve the previously mentioned bound. During our analysis  interesting connections  to the \emph{Dedekind's zeta-function} and to \emph{Dirichlet's unit theorem}   are revealed. Finally we prove that both the number field codes and Alamouti code are in some sense optimal codes in the multiple access channel (MAC).
\end{abstract}

\section{Introduction}

In \cite{MACDMT} the authors gave diversity multiplexing trade-off for  MIMO (multiple-input multiple-output) MAC. 
  In their paper Tse, Viswanath and Zheng pointed out that the MAC-DMT is obviously always upper bounded by the DMT of the single-user.    In this paper we are concentrating on the scenario where the single-user codes are not DMT optimal. In such a scenario it is obvious that we cannot achieve the optimal MAC DMT given in  \cite{MACDMT}.  However we can ask another question: in which cases the single-users can maintain their single-user DMT- performance despite the interference of the other users.  
  
The importance of this problem lies in the fact that in some scenarios  codes achieving the optimal DMT can have  high decoding complexity.
As an example of a scheme, let us  consider  the situation where we have two users, both using Alamouti \cite{Alam} code, and where the receiver has two antennas.  The decoding complexity of this coding scheme  is still relatively light to decode even when for example sphere decoding is used.

In this  case the  DMT of a single-user can never be better than the performance of the Alamouti code in the
$2\times2$ MIMO channel. Therefore we are immediately bounded away from the optimal achievable MAC DMT. However, we can ask whether both transmitters can achieve their single-user performance despite the interference of the other user. 

\emph{Throughout the paper we are considering the symmetric multiplexing gain scenario}

\section{Basic definitions}
Let us now consider a slow fading channel where we have $n_t$ transmit and $n_r$  receiving antennas and  where the decoding delay is $T$ time units. 
The channel equation can be now written as
$$
Y=\sqrt{\frac{SNR}{n_t}}HX +N
$$
where $H \in M_{n_r \times n_t}(\C)$ is the channel matrix whose entries are independent identically
distributed (i.i.d.) zero-mean complex circular Gaussian
random variables with the variance 1, and $N\in M_{n_r \times T}(\C) $
is the noise matrix whose entries are i.i.d. zero-mean complex circular Gaussian random variables with the variance 1.
Here $X \in M_{n_t\times T}(\C)$ is the transmitted codeword and $SNR$ presents the signal to noise ratio.

In order to shorten the notation we denote $SNR$ with $\rho$.
Let us suppose we have coding scheme where for  each value of $\rho$ we have a code $C(\rho)$ having
$|C(\rho)|$ matrices in $M_{n \times T}(\C)$. The rate $R(\rho)$ is then $\log{(|C(\rho))|}/T$.
Let us suppose  that the scheme fulfills the power constraint  
\begin{equation}\label{powerconstraint}
\frac{1}{|C(\rho)|}\sum_{X \in C(\rho)} ||X||_F^2 \leq T n_t.
\end{equation}

We then have the following definition from \cite{ZT}.
\begin{definition}
The scheme is said to achieve \emph{spatial multiplexing gain} $r$ and \emph{diversity gain} $d$ if the data rate
$$
\lim_{\rho \to \infty} \frac{R(\rho)}{log(\rho)} = r
$$
and  the average error probability
$$
\lim_{\rho \to \infty} \frac{log(P_e(\rho))}{log(\rho)}=-d.
$$

\end{definition}

Let us now consider a coding scheme based on a $k$-dimensional lattice $L$ inside $M_{n\times T}(\C)$ where for a given positive real number $R$ the finite code is
$$
L(R)=\{a|a \in L,||a||_F\leq R \}.
$$
The following lemma is a  well known result from basic lattice theory. 

\begin{lemma}\label{spherical}
Let $L$ be a  $k$-dimensional lattice in  $M_{n\times T}(\C)$ and
$$L(R)=\{a \,|\, a\in L, \, ||a||_F\leq R \,\},$$
then
$$
|L(R)|= cR^{k}+ f(R),
$$
where $c$ is some real constant and $|f(R)| \in o(R^{(k-1/2)})$. 
\end{lemma}

In particular it follows that we can choose real numbers $K_1$ and $K_2$ so that
$$
K_1R^k\geq |L(R)|\geq K_2 R^k.
$$

If we then consider a coding scheme where the finite codes are sets
\begin{equation}\label{codingscheme}
C_L( \rho^{rT/k})=\rho^{-rT/k}L(\rho^{rT/k}),
\end{equation}
we will get a  correct number of codewords for each $\rho$ level and the sets $C_L (\rho^{rT/k})$ clearly do fulfill the average energy constraints  \eqref{powerconstraint}. Here and in the following we simply forget the term $\frac{1}{n_t}$ in the channel equation as it is irrelevant in DMT calculations. 

\section{An upperbound for the DMT of a $2n$-dimensional lattice code in $M_n(\C)$}
In this section we are going to give a simple analysis of achievable DMT of a $2n$-dimensional lattice code in $M_n(\C)$. 

The following lemma is a simple corollary to the Lemma \ref{spherical}.
\begin{lemma}\label{smallradius}
Let us suppose  that we have a $2n$-dimensional lattice $L$ in $M_n(\C)$ and a positive constant $k$. We then have   constants $K$ and $M$  independent of $R$ such that
$$
KR^{2n}\geq|\{X \,|\, X \in L, \, ||X||_F \leq R-k   \}|\geq MR^{2n}.
$$

\end{lemma}

Let  $P_e(\rho, X\rightarrow X')$ denote the  error probability of decoding $X'$ when $X$, was transmitted at SNR $\rho$.
\begin{proposition}\cite{LuKumar}\label{PEPlower}
Let us suppose that we have two codewords $X, X' \in M_n(\C)$,  and that $det(X-X')\neq 0$.
We  then have that
$$
P_e (\rho, X\rightarrow X')\geq \rho ^{-nn_r}K|det(X-X')|^{-2n_r},
$$
for some constant $K$ independent of $\rho$ (but not independent of $X$ and $X'$).
\end{proposition}

Let us now consider the previously defined spherical coding scheme and a $2n$-dimensional lattice code $L\subseteq M_n(\C)$.  $C_L(\rho^{r/2})=\rho^{r/2}L(\rho^{r/2})$. 

\begin{proposition}\label{simplelower}
Let us suppose that we have a $2n$-dimensional lattice  code $L$ in $M_{n\times n}(\C)$. If the transmitter has $n$ antennas and the receiver has $n_r$ antennas the DMT of this code is then upper bounded by the curve
$$
(r, n_rn(1-r)).
$$
\end{proposition}

\begin{proof}
We now have that the average error probability is
$$
P_e= \frac{1}{|L(\rho^{r/2})|}\sum_{X\in L(\rho^{r/2})} P_e(\rho,\rho^{-r/2}X),
$$
where  $P_e(\rho ,\rho^{-r/2}X)$ is the average probability for making a  mistake in receiving if  a codeword $\rho^{-r/2}X \in \rho^{-r/2}L(\rho^{r/2}) $ was transmitted with SNR $\rho$.

Let us choose such an $X_{min} \in L$  that $||X_{min}||_F=M$ is the smallest possible. We can always suppose that it is included to the set
$L(\rho^{r/2})$.

Let us now use the following notation
$$
 P_e(\rho,\rho^{-r/2}X)=Y_X.
$$
We can now divide the average error probability into two parts
$$
P_e= \frac{1}{|L(\rho^{r/2})|}\sum_{X\in L(\rho^{r/2})/L(\rho^{r/2}-M)  } Y_X +
$$
$$
\frac{1}{|L(\rho^{r/2})|} \sum_{X\in L(\rho^{r/2}-M)} Y_X,
$$
where $L(\rho^{r/2})/L(\rho^{r/2}-M)$ refers to difference of sets.  

We will now prove that for all the elements $X$ in $L(\rho^{r/2}-M)$, $Y_X$ is "large" and we prove that the number of $X$ in  $L(\rho^{r/2}-M)$ is a large enough compared to $|L(\rho^{r/2})|$. Therefore the contribution of these $Y_X$ is enough to give a lower bound  for the error probability.

Let us now consider  $Y_{X_i}$ where $X_i \in L(\rho^{r/2}-M)$. We remark that if $X_i \in L(\rho^{r/2}-M)$, then
$X_i+X_{min} \in L(\rho^{r/2})$, which follows from the triangle inequality
$$
||X_{min}+X_i||_F\leq || X_{min}||_F+||X_i||_F \leq M + \rho^{r/2}-M.
$$

The average error probability is always  bounded by  any pairwise error probability. Therefore for
$X_i \in L(\rho^{r/2}-M)$ we have
$
Y_{X_i}\geq P_e(\rho,  \rho^{-r/2}X_i  \rightarrow \rho^{-r/2}(X_i +X_{min})).
$

According to Proposition \ref{PEPlower} we get that
$$
P_e(\rho, \rho^{-r/2} X_i \rightarrow \rho^{-r/2}(X_i +X_{min})) 
$$
$$
\geq K_1 \rho^{-nn_r(1-r)}| \det(X_{min}^{\dagger} X_{min}) |^{-n_r }
$$
where $K_1$ is a constant independent of $\rho$.
It follows that
$$
\frac{1}{|(L(\rho^{r/2})|}\sum_{X \in L(\rho^{r/2}-M)} Y_X \geq |( L(\rho^{r/2}-M)| 
$$
$$
\cdot\frac{1}{|(L(\rho^{r/2})|}  K_1 \rho^{-n n_r(1-r)}| \det(X_{min}^{\dagger} X_{min})|^{-n_r }.
$$

According to Lemma \ref{spherical} there does exist such a constant $K_2$ that  $|(L(\rho^{r/2})| \leq K_2\rho^{rn}$ and from Lemma \ref{smallradius}
it follows that there is such a constant $K_3$ that $|( L(\rho^{r/2}-M)|\geq K_3\rho^{rn}. $
Combining these and the previous we have

$$
P_e\geq  \frac{1}{|L(\rho^{r/2})|} \sum_{X\in L(\rho^{r/2}-k)} Y_X
$$
$$
\geq K_2^{-1}\rho^{-rn}\cdot K_1 \rho^{-n n_r(1-r)}| \det(X_{min}^{\dagger}X_{min})|^{-n_r }K_3 \rho^{rn}
$$

$$
\geq M \rho ^{-nn_r(1-r)},
$$
where $M$ is a constant independent of $\rho$.
\end{proof}

\section{A union  bound based DMT analysis of some MISO codes}
In this section we are giving union bound based proofs for the DMT of Alamouti code and number field  based codes \cite{Belfiore}. While our approach is more laborious than the proofs usually given, it will later be proved to be helpful in MAC scenario. We point out that the achieved DMT's do achieve the bound \ref{simplelower}.
\subsection{Alamouti code}
Let us warm up by calculating the DMT of the Alamouti code in the case where  we have $n_r$ receiving antennas.
Let us use the following notation
$$
A(x_1, x_2, x_3, x_4)=
\begin{pmatrix}
x_1+x_2i& -(x_3 +x_4i)^*\\
x_3+x_4i& (x_1+x_2i)^*
\end{pmatrix}.
$$
We then have the following
$$
A(x_1,x_2, x_3,x_4)A(x_1, x_2, x_3, x_4)^{\dagger}=
$$
$$
\begin{pmatrix}
x_1^2+x_2^2 +x_3^2+x_4^2&0\\
0&x_1^2+x_2^2 +x_3^2+x_4^2
\end{pmatrix}.
$$

Here the lattice $L$ is 
$$
\Z A(1,0,0,0)+\Z A(0,1,0,0)+ \Z A(0,0,1,0)+\Z A(0,0,0,1),
$$
which  is a $4$-dimensional lattice in $M_2(\C)$. For simplicity we do not use the spherical shaping scheme, but instead  we consider
the following scheme
$$
C_1(\rho^{r/2})=\{\rho^{-r/2} A(x_1, x_2, x_3,x_4)| -\rho^{r/2}\leq x_i\leq \rho^{r/2}\},
$$
where $x_i\in\Z$.
\begin{proposition}\label{AlamDMT}
 When received with $n_r$  antennas the Alamouti code achieves the DMT curve
$$
(r,  2n_r(1-r)).
$$

\end{proposition} 

\begin{proof}
 The usual union bound argument now gives us the following bound for the error probability of making a mistake in reception when  transmitting arbitrary codeword
$$
 P_e \leq \sum_{ -2\rho^{r/2} \leq x_i \leq 2\rho^{r/2}, x_i \in \Z} \frac{\rho^{-2n_r(1-r)}}{(det(A(x_1, x_2, x_3,x_4))^{2n_r}}
 $$
$$
=\sum_{ |x_i| \leq 2\rho^{r/2}, x_i\in \Z} \frac{\rho^{-2n_r(1-r)}}{(x_1^2+x_2^2 +x_3^2+ x_4^2)^{2n_r}},
$$
where we suppose that not all $x_i$ can be $0$ at the same time.
If we then apply AM-GM inequality, we get the following
$$
P_e\leq \sum_{|x_i|\leq 2\rho^{r/2}, x_i \in \Z } \frac{\rho^{-2n_r(1-r)}}{|\stackrel{.}{x_1}\stackrel{.}{x_2}\stackrel{.}{x_3}\stackrel{.}{x_4}|^{n_r}},
$$
where the dot sign means that if $x_i=0$ we have that $\stackrel{.}{x_i}=1$.
By considering the right side of the previous equation  we  have that
$$
P_e\leq \rho^{-2n_r(1-r)}\left(\sum_{|x_1| \leq 2\rho^{r/2}, x_i  \in \Z }\frac{1}{|\stackrel{.}{x_1}|^{n_r}}\right) \cdot
$$
$$
\cdots \left(\sum_{|x_4| \leq 2\rho^{r/2}, x_i \in \Z }\frac{1}{|\stackrel{.}{x_4}|^{n_r}}\right)
$$
$$
\leq \rho^{-2n_r(1-r)}K(2log(2\rho^{r/2}))^{4n_r},
$$
where $K$ is some constant independent of $\rho$.
\end{proof}

\subsection{Diagonal Number field codes}

For simplicity let us consider a degree $n$ cyclic number field extension $K/\Q(i)$, where the Galois group is $<\sigma>$. Then we can define a
\emph{relative canonical embedding} of $K$ into $M_n(\C)$ by
$$
\psi(x)=\mathrm{diag}(\sigma_1(x),\dots, \sigma_n(x)),
$$
where $x$ is an element in $K$.
The ring of algebraic integers $\OO_K$ has a  $\Z$-basis $W=\{w_1,\dots ,w_{2n}\}$ and therefore
$$
\psi(\OO_K)=\psi(w_1)\Z+\cdots +\psi(w_{2n})\Z,
$$
is a $2n$-dimensional lattice of matrices in $M_n(\C)$. The main reason to use such a code construction is that for each  nonzero $a\in \OO_K$, we have that $|det(\psi(a))|\geq 1$. Let us now suppose that we have  a $2n$-dimensional  number field lattice code $L \subseteq M_n(\C)$ and that we are considering the coding scheme, where the finite codes are chosen by  the method of Lemma \ref{spherical}. 

We will now measure the DMT  of these type of codes. Before that we will need some concepts and lemmas.

 The unit group $U_K$ of the ring
$\OO_K$ consists of such elements $u \in \OO_K$, that $|\mathrm{det}(\psi(u))|=1$.

 We skip the proof of the following lemma, which is a corollary to \emph{Dirichlet's unit theorem}.
\begin{lemma}\label{units}
Let us suppose that we have a cyclic extension $K/\Q(i)$, where $[K:\Q(i)]=n$. 
Let us now consider the set
$$
U_K(R)=\{\psi(u)|\, u\in U_K, \, ||\psi(u)||_F\leq R \,\},
$$
we then have that 
$$
|U_K(R)|\leq M log(R)^{n-1},
$$
where $M$ is a constant independent of $R$.
\end{lemma}

\begin{corollary}\label{Ux}
Let us suppose that we have a cyclic extension $K/\Q(i)$, where $[K:\Q(i)]=n$. Let us suppose we have  a non-zero element $x\in \OO_K$, where $||\psi(x)||_F\leq R$. We then have that
$$
|\psi(U_Kx) \cap B(R)|=|\{u \,|\, ||\psi(xu)||_F \leq R, u \in U_K \}| 
$$
$$
\leq Mlog(R)^{n-1},
$$
where $M$ is a constant independent of $R$ and of the element $x$.
\end{corollary}
\begin{proof}
We can write $\psi(x)=diag(x_1,\dots, x_n)$.
The condition $||\psi(x)||_F\leq R$ then gives us that  $|x_i|\leq R \, \, \forall i$. We also have that 
$ |x_1|\cdots |x_n|\geq 1$. It now follows that
\begin{equation}\label{coordinatesize}
 |x_i|\geq \frac{1}{R^{n-1}}\, \forall i.
 \end{equation}
 Let us now suppose that $u$ is such  a unit  that
$||\psi(ux)||_F=||\psi(u)\psi(x)||_F=||diag(x_1 u_1,\dots, x_nu_n)||_F\leq R$. Equation \eqref{coordinatesize} now gives us that
 that $ |u_i|  \leq R^n \, \forall i$. Therefore we have  that $||\psi(u)||_F \leq \sqrt{n}R^n$. Lemma \ref{units} now gives us that  $|U_K x(R)\cap B(R)|\leq  M log(\sqrt{n}R^n)^{n -1}\leq M_1 log(R)^{n -1}$, where $M_1$ is a constant independent of $R$.

\end{proof}

In the following we will use the term $\mathbf{I}_K$ for the set of integral ideals of the ring $\OO_K$.

\begin{proposition}\label{zeta}
Let us suppose that we have a cyclic  extension $K/\Q(i)$, where $[K:\Q(i)]=n$. If $\OO_K$ is principal ideal domain (PID) we have the following
$$
\sum_{||\psi(x)||_F\leq R, x\in X} \frac{1}{|det(\psi(x))|^{2n_r}} \leq M log(R)^{2n},
$$
where $X$ is such a set of elements $x$, $||\psi(x)||_F \leq R$, of $\OO_K$ that each generate a separate integral ideal.
\end{proposition}
\begin{proof}
Using basic properties of algebraic norm and AM-GM inequality we have the following 
$$
|det(\psi(x))|^2=|N_{K/\Q}(x)|\leq  ||\psi(x)||_F^{2n},
$$
for any element in $\OO_K$. 
This gives us that
$$
\sum_{\substack{||\psi(x)||_F\leq R\\ x \in \OO_K}} \frac{1}{|\mathrm{det} (\psi(x))|^{2n_r}}= \sum_{\substack{|N_{K/\Q}(x)|\leq R^n \\ ||\psi(x)||_F\leq R  } } \frac{1}{|N_{K/\Q}(x)|^{n_r}},
$$
where we sum only over a set of elements each generating a separate integral ideal.  Due to this limitation and relation between ideal and element norms  we have the following
$$
\sum_{\substack{|N_{K/\Q}(x)|\leq R^n \\ ||\psi(x)||_F\leq R  } } \frac{1}{|N_{K/\Q}(x)|^{n_r}}\leq \sum_{\substack{|N_{K/\Q}(I)|\leq R^n\\I\in \textbf{I}_K}} \frac{1}{|N_{K/\Q}(I)|^{n_r}},
$$
where $I$ represents and integral ideal.
But this is  the beginning of the Dedekind's zeta-function at point $n_r$! We then have the following
$$
\sum_{ |N_{K/\Q}(I)|\leq R^n, I\in \textbf{I}_K} \frac{1}{|N_{K/\Q}(I)|^{n_r}}
$$
$$
 \leq \left( \sum_{i< R^n, i\in \Z^+ } \frac{1}{i^{n_r}}\right)^{2n}\leq(log(R^n))^{2n},
$$
where the first inequality is based on similar reasoning as in \cite[Prop. 7.2, Cor. 3]{Nark} and the last one is based on elementary approximation.

\end{proof}
Note that if $n_r>1$, Proposition \ref{zeta} gives tighter bound.
\begin{proposition}\label{numberfieldsum}
Let us suppose we have cyclic degree $n$ extension $K /\Q(i)$,  and that $\OO_K$ is a principal ideal domain. We then have that
$$
\sum_{||\psi(a)||_F\leq R, a\in \OO_K }\frac{1}{|det(\psi(a))|^{2n_r}}\leq Mlog(R)^{3n-1}.
$$
\end{proposition}
\begin{proof}
Just as in the proof of Proposition \ref{zeta} we can write
$$ 
\sum_{||\psi(a)||_F\leq R, a\in \OO_K }\frac{1}{|det(\psi(a))|^{2n_r}}
$$
$$
=\sum_{\substack{||\psi(a)||_F\leq R \\ |N_{K/\Q}(a))|\leq R^n, a\in \OO_K }}\frac{1}{|N_{K/\Q}(a)|^{n_r}}.
$$
The right side can then been written as
$$
\sum_{\substack{|N_{K/\Q}(x_i)|\leq R^{n}, \\x_i\in X}}\frac{A_i}{|N_{K/\Q}(x_i)|^{n_r}},
$$
where $X$ is some collection of elements $x_i \in \OO_K$, $||x_i||_F\leq R$, such that each generate separate integral ideal.  The numbers  $A_i$ present the number of elements inside $B(R)$ each generating the same ideal $x_i\OO_K$. 
As we supposed that  $\OO_K$ is a PID  Lemma \ref{zeta} gives us that 
$$
\sum_{|N_{K/\Q}(x_i)|\leq R^n, x_i\in X }\frac{1}{|N_{K/\Q}(x_i)|^{n_r}}\leq M_1 log(R)^{2n}.
$$

 From the ideal theory we know that if $x_k\OO_K=x_k'\OO_K$, then $x_k$ and $x_k'$ must differ by a unit. Therefore  we can now apply Lemma \ref{Ux} that gives us that  for all $A_i$ we have $A_i \leq M_2log(R)^{n-1} $. Combining now this  and Proposition \ref{zeta} we have
$$
\sum_{\substack{|N_{K/\Q}(x_i)|\leq R^{n}, \\x_i\in X }}\frac{A_i}{|N_{K/\Q}(x_i)|^{n_r}}\leq   M_1 M_2 log(R)^{n-1} log R^{2n}
$$
$$
=Mlog(R)^{3n-1}.
$$
The crucial point here was that we could choose the constant $M_2$ so that it bounds every $A_i$.

\end{proof}

Let us now consider a number field code $L\subset M_n(\C)$ and use the spherical coding scheme  \eqref{codingscheme}.

\begin{corollary}\label{numberfieldsDMT}
Let us suppose that we have a previously described number field code $L\subset M_n(\C)$. If the receiver has $n_r$  antennas we  achieve the  DMT curve
$$
(r, n n_r(1-r)).
$$

\end{corollary}
\begin{proof}
The code lattice $L\subseteq M_n(\C)$ has dimension $2n$. 
The finite codes attached to the spherical coding scheme are then 
$$
C_L(\rho^{r/2})=\rho^{-r/2}L(\rho^{r/2}).
$$

By the usual union bound argument we have the following upper bound for the  average error probability
$$
P_e \leq \sum_{X\in C_L(2\rho^{r/2})} \frac{\rho^{-nn_r(1-r)}}{|det(X)|^{2n_r}},
$$ 
where we have used the knowledge of the lattice structure of the  code $L$. In order to take into account that we are considering differences between
codewords we also took the sum over a ball with double radius.

Just as previously we have
$$
\sum_{X\in L(2\rho^{r/2})} \frac{\rho^{-n_r n(1-r)}}{|det(X)|^{2n_r}} 
$$
$$
=\sum_{||\psi(a)||_F\leq 2\rho^{r/2}, a\in \OO_K }\frac{\rho^{-n_r n(1-r)}}{|det(\psi(a)|^{2n_r}}.
$$
According to Proposition \ref{numberfieldsum} we now have
$$
\sum_{X\in L(2\rho^{r/2})} \frac{\rho^{-n_r n(1-r)}}{|det(X)|^{2n_r}} 
\leq \rho^{-n_r n(1-r)}log(2\rho^{r/2})^{3n-1}.
$$

\end{proof}

\section{MISO codes in MAC scenario}
Let us now consider a scenario where we have $K$ independent users  each using  $2n$-dimensional MISO-lattice codes $L_1,\dots, L_K \subseteq M_n(\C)$ and that the receiver has $n_r\geq K$ antennas. In this section we prove that if each user uses a MISO code from the previous sections (Alamouti or number field code)  they can reach the single-user DMT despite the interference of the other users. According to Proposition \ref{simplelower} the achieved DMT:s are the best  it is possible to get when the users are applying $2n$-dimensional lattice codes.

\begin{lemma}
The product of singular values (non-zero)
of the matrix  $AA^{\dagger}$ are the same as those of
$A^{\dagger}A$.
\end{lemma}

The following result is well known from matrix theory.
\begin{lemma}\label{sumanproduct}
Let us consider a  $Kn\times n$ matrix $X=[X_1,\dots, X_K]^T$. We then have that
$$
\mathrm{det}((X)(X)^{\dagger}) \geq \sum_{i=1}^{K} \mathrm{det}(X_iX_i^{\dagger}).
$$
\end{lemma}

 Let us  suppose that the receiver uses  joint decoding.
 As noted in \cite{MACDMT} this choice of receiving strategy does not change the DMT performance of each user.
We can now consider  the whole system as a single-user code where the single-user has $Kn$ transmit antennas and the receiver has 
$n_r$ receiving antennas.   The single-user code can then be defined as
$$
L= \{[X_1, X_2,\dots, X_K]^T \,|\,  X_i \in L_i \,\} \subseteq M_{Kn\times n}(\C).
$$
As each of the lattices $L_i$ are $2n$-dimensional the lattice $L$ is $2Kn$-dimensional.

Following the previously defined coding scheme \eqref{codingscheme} we define the finite codes needed in DMT analysis by
$$
C_L(\rho^{r/2K})=\rho^{-r/2K}L(\rho^{r/2K}).
$$
Let us now suppose that each $L_i$ is either number field code as defined previously or in the case $n=2$ Alamouti code.

The crucial properties of the  codes $L_i$ are the following.
\begin{itemize}
\item We have $|det(X_i)|\geq 1$, when $X_i \neq 0$ and $X_i\in L_i$.
\item We also have the inequality 
\begin{equation}\label{generalsum}
\sum_{ X \in L_i |\,||X||_F\leq R}\frac{1}{|det(X)|^2}\leq Slog(R)^M, 
\end{equation}
where  $S$  and $M$ are some constants.
\end{itemize}
\begin{proposition}\label{macnumberfield}
Let us suppose that we have the previously described coding scheme and that the receiver has $n_r\geq K$ antennas.
Then the code $L$ achieves the DMT curve
$$
(r, n n_r(1-r/K))
$$
and  each single-user achieves DMT curve
$$
(r, n n_r(1-r)).
$$
\end{proposition}

\begin{proof}
The usual union bound argument gives us
$$
P_e\leq \rho^{-(nn_r(1-r/K))}\sum_{X\in L(2\rho^{r/2K})} \frac{1}{|det(X X^{\dagger})|^{n_r}} 
$$
$$
\leq \rho^{-(nn_r(1-r/K))} \sum_{X_i \in L(2\rho^{r/2K})} \frac{1}{(|det(X_1)|^2+ \cdots +det|(X_K)|^2)^{n_r}}
$$
$$
\leq \rho^{-nn_r(1-r/K)} \sum_{X_i \in L_i(2\rho^{r/2K})} \frac{1}{(\stackrel{.}{|det(X_1)|}\cdots\stackrel{.}{|det(X_K)|})^{2n_r/K}}
$$
 
$$
\leq \rho^{-nn_R(1-r/K)} \left(\sum_{X_1  \in L_1(2\rho^{r/2K})} \frac{1}{|det(X_1)|^2} \right)^{n_r}\cdots
$$
 $$
\cdots \left(\sum_{X_K L_K(2\rho^{r/2K})} \frac{1}{|det(X_K)|^2}\right)^{n_r}
$$
$$
\leq S\rho^{-nn_r(1-r/K)}  log(2\rho^{r/2K})^ {VK},
$$
where the last inequality comes from  condition \eqref{numberfieldsum} and where the dotted notation $\stackrel{.}{det(X_i)}$=1, if $X_i=0$. Here $S$ and $V$ are constants.
This finally gives us 
$$
P_e\leq M\rho^{-nn_r(1-r/K)}log(\rho)^T,
$$
where  $M$ and $T$ are some constants independent of $\rho$. This gives us the first claim. 

In order to get the single-user perspective we have to multiply the multiplexing gain $r$ with number of user $K$. This gives us
the second claim.
\end{proof}

\begin{remark}
We  point out that in the previous result we could combine any kind of codes as long  as they fulfilled the conditions given before Proposition \ref{macnumberfield}.
The DMT of stacked Alamouti codes was earlier proved in \cite{stackedDMT}.

\end{remark}

\section*{Acknowledgement}
 The research of  R. Vehkalahti is supported by the Emil Aaltonen Foundation and by the Academy of
Finland (grant 131745). 

\end{document}